\documentclass[letterpaper,10pt,conference]{ieeeconf}

\usepackage{cite}
\usepackage{graphicx}
\usepackage{subcaption}
\usepackage[cmex10]{amsmath}
\usepackage{algpseudocode}
\usepackage{algorithmicx}
\usepackage{xcolor}
\usepackage{array}
\usepackage{url}
\usepackage{bbm}
\usepackage{todonotes}

\usepackage{epsfig}
\usepackage{amsfonts,amssymb,multirow,bigstrut,booktabs,ctable,latexsym}
\usepackage{mathtools}
\usepackage[mathscr]{eucal}
\usepackage{verbatim}

\usepackage{amsthm}
\usepackage{algorithm}
\usepackage[normalem]{ulem}
\newtheorem{definition}{Definition}
\newtheorem{lemma}{Lemma}
\newtheorem{theorem}{Theorem}

\newtheorem{proposition}{Proposition}

\newtheorem{assumption}{Assumption}

\newtheorem{problem}{Problem}

\IEEEoverridecommandlockouts
\begin{document}
	
	\title{\bf Safety-Critical Online Control with Adversarial Disturbances}
	
	\author{Bhaskar Ramasubramanian$^{1}$, Baicen Xiao$^{1}$, Linda Bushnell$^{1}$, and Radha Poovendran$^{1}$%
		\thanks{$^{1}$Network Security Lab, Department of Electrical and Computer Engineering, 
			University of Washington, Seattle, WA 98195, USA. \newline
			{\tt\small \{bhaskarr, bcxiao, lb2, rp3\}@uw.edu}}
		\thanks{This work was supported by the U.S. Army Research Office and the Office of Naval Research via Grants W911NF-16-1-0485 and N00014-17-S-B001 respectively. }
	}
	
\maketitle
	
\begin{abstract}
This paper studies the control of safety-critical dynamical systems in the presence of adversarial disturbances. 
We seek to synthesize state-feedback controllers to minimize a cost incurred due to the disturbance while respecting a safety constraint. 
The safety constraint is given by a bound on an $\mathcal{H}_\infty$ norm, while the cost is specified as an upper bound on an $\mathcal{H}_2$ norm of the system. 
We consider an online setting where costs at each time are revealed only after the controller at that time is chosen. 
We propose an iterative approach to the synthesis of the controller by solving a modified discrete-time Riccati equation. 
Solutions of this equation enforce the safety constraint. 
We compare the cost of this controller with that of the optimal controller when one has complete knowledge of disturbances and costs in hindsight. 
We show that the \emph{regret} function, which is defined as the difference between these costs, varies logarithmically with the time horizon. 
We validate our approach on a process control setup that is subject to two kinds of adversarial attacks.
\end{abstract}
	
\section{Introduction}

The recent advances and successes of reinforcement learning (RL) \cite{sutton2018reinforcement} in 
robotics, games, and mobile networks \cite{hafner2011reinforcement, mnih2015human, lillicrap2016continuous, silver2016mastering, zhang2019deep} has spurred its use in other areas where RL algorithms interact with the physical environment over long periods of time  \cite{sadigh2016planning, yan2018data, you2019advanced}. 
An increasingly popular domain where RL methods are being deployed are to  
safety-critical systems like large-scale power systems \cite{lewis2012reinforcement},  
which are susceptible to attacks by an intelligent adversary \cite{banerjee2012ensuring, sullivan2017cyber}. 
Since these systems have an underlying dynamic model, actions of the system are typically a function of (a history of) the system states. 
Rules governing these actions can be designed so that the overall system behaves in a desired way. 
At the same time, the actions may have to be chosen to minimize a cost. 

We consider a safety-critical linear time invariant (LTI) system affected by adversarial inputs. 
Our goal is to design state-feedback controllers to minimize the cost incurred due to this input while satisfying a safety constraint. 
This is also called the `\emph{combined $\mathcal{H}_2/ \mathcal{H}_\infty$ problem}' \cite{haddad1991mixed, mustafa1991lqg, kaminer1993mixed}. 
The $\mathcal{H}_\infty$ safety constraint enforces a bound on the ratio of the magnitude of the output to that of the adversarial disturbance. 
The $\mathcal{H}_2$ cost is the expected mean square output value when the disturbance input is a white noise process. 
The $\mathcal{H}_\infty$ constraint is embedded in the optimization process by solving a modified discrete-time Riccati equation, whose solution yields an upper bound on the $\mathcal{H}_2$ cost.  

In this paper, we study an online scenario of the combined $\mathcal{H}_2/ \mathcal{H}_\infty$ problem. 
At each time, the adversary inserts a disturbance, after which the cost incurred is revealed to the system. 
Our aim is to iteratively design controllers to minimize this cost, while satisfying the safety constraint. 
We compare the cost of this controller with that of the optimal controller if all adversarial inputs and costs were known apriori. 
The difference between these costs is termed as the \emph{regret} faced by the system. 

Regret bounds for partially and fully adversarial disturbances for LTI dynamics and convex costs were presented in \cite{agarwal2019logarithmic, agarwal2019online, hazan2019nonstochastic, simchowitz2020improper, foster2020logarithmic},  
where the authors considered a richer class of `disturbance action policies'. 
These policies depend not only on the current state, but also on a history of disturbances, and provide stronger regret bounds (poly-logarithmic) than we present (logarithmic), since policies in this paper only depend on the current state. 
Also, while the analyses in \cite{agarwal2019logarithmic, agarwal2019online, hazan2019nonstochastic, simchowitz2020improper, foster2020logarithmic} fix a stabilizing controller at the start of their algorithms, we adopt a different approach and iteratively solve a set of Riccati equations to update the controller at each time step.
\subsection{Contributions}

We aim to minimize an upper bound on the $\mathcal{H}_2$ cost for an LTI system with adversarial inputs with an $\mathcal{H}_\infty$ constraint on the disturbance-output map in an online setting. 
At each time, the adversary inserts a disturbance input. 
The cost functions are revealed to the system only 
after it has determined a controller. 
We make the following contributions:
\begin{itemize}
\item We introduce strongly stable disturbance attenuating (S2DA) policies. This generalizes strongly stable policies from \cite{cohen2018online}. S2DA policies are strongly stable and satisfy an additional condition on the $\mathcal{H}_\infty$ norm. 
\item We show that initializing our procedure with a stabilizing disturbance attenuating policy will yield stabilizing disturbance attenuating policies at successive time steps. 
If solutions of the Riccati recursion are bounded, we show that these policies will also be strongly stable.
\item We establish bounds on the difference between solutions to the Riccati equation at successive time steps. 
We use the above results to show that the regret bound is $O(\log T)$, where $T$ is the time horizon of interest.
\item We validate our method on a model of the Tennessee Eastman control challenge \cite{downs1993plant} that is subject to arbitrary adversarial inputs, and a denial-of-service attack. 
\end{itemize}
\subsection{Outline of Paper}

The rest of this paper is organized as follows: 
Section \ref{RelatedWork} summarizes related work. 
We state our problem and detail our solution in Sections \ref{Problem} and \ref{Solution}. 
Section \ref{Simulation} illustrates our approach on a model of the Tennessee Eastman control challenge. Section \ref{Conclusion} concludes the paper. 

\section{Related Work}\label{RelatedWork}

Simultaneous $\mathcal{H}_2/ \mathcal{H}_\infty$ policy synthesis for discrete-time linear systems is a well-studied problem. 
The structure of an upper bound for the LQG cost, minimizing which would solve the mixed problem was first proposed in \cite{haddad1991mixed}. 
The authors of this paper also presented a closed-form controller that would minimize this upper bound. 
Two other upper bounds for the LQG cost were proposed in \cite{mustafa1991lqg}, and it was shown that the same controller would be optimal in each case when restricted to static, full-state feedback. 
In this case, a static, time-invariant state-feedback is sufficient for optimal performance \cite{kaminer1993mixed}. 
In discrete-time, this does not hold for the full-information feedback (states and disturbances available), or partial information cases. 
This problem was studied for nonlinear discrete-time systems in \cite{aliyu2008discrete}. 
An orthogonal approach to solve disturbance attenuation and rejection problems using geometric control theory was presented in \cite{wonham1974linear, willems1981almost, basile1992controlled, furuta1993closed, saberi1996h2, lin2000solutions}. 
However, these works considered the $\mathcal{H}_2$ and $\mathcal{H}_{\infty}$ cases separately. 
These problems have also been studied in robust and model predictive control \cite{zhou1996robust, bemporad1999robust, rakovic2018handbook}.

There has been a renewed interest in the use of RL techniques in learning to control linear dynamical systems. 
Recent developments in this field are surveyed in \cite{recht2019tour}. 
An online version of LQ control with Gaussian disturbances was presented in \cite{cohen2018online}, where the authors presented a regret bound for known LTI dynamics and adversarial quadratic costs. 
In \cite{agarwal2019logarithmic}, the authors considered strongly convex costs, and presented stronger regret bounds for a larger class of policies they termed disturbance action policies. 
This was generalized to semi-adversarial disturbance inputs and convex costs in \cite{agarwal2019online}. 
The authors of \cite{akbari2019iterative} adopted a different approach to determine regret bounds for the LQR. 
They iteratively solved a sequence of Riccati equations to generate a sequence of stabilizing controllers, and compared the cost of this sequence of controllers with that of the optimal static controller if all costs were known in hindsight. 

Minimizing the regret over sequentially revealed adversarial convex costs against the class of linear policies when a model of the dynamics was unknown was studied in \cite{hazan2019nonstochastic}. 
This was generalized to partially observed systems with semi-adversarial disturbances in \cite{simchowitz2020improper}, and for the Kalman filter in \cite{tsiamis2020online}. 
More recently, \cite{foster2020logarithmic} presented regret bounds for the case of fully adversarial disturbances. 

Sample complexity bounds for the LQR for an LTI system with unknown dynamics were given in \cite{dean2019sample}, and for the Kalman filter in \cite{tsiamis2019sample}. 
The convergence of policy-gradient methods for the LQR was studied in \cite{fazel2018global, tu2019gap, gravell2019learning}. 
Convergence guarantees for policy gradient methods for the mixed $\mathcal{H}_2/ \mathcal{H}_\infty$ problem was reported in \cite{zhang2019policy}. 
The authors of \cite{dean2019safely} studied a trade-off between exploration for learning and safety for the LQR under bounded disturbances and constraints on the state and input sets. 

\section{Preliminaries and Problem Formulation}\label{Problem}

For a matrix $M$, we write $M \geq 0$ when $M$ is positive semi-definite. 
We write $M_1 \geq M_2$ when $M_1 - M_2 \geq 0$. 
$Tr(M)$ and $\lambda_{max} (M)$ denote the trace and maximum eigenvalue of $M$. 
Consider a discrete-time linear system:
\begin{align}
x_{t+1}&=Ax_t+Bu_t+Dw_t \label{StateEqn}\\
z_t&=C x_t + E u_t \label{OutputEqn}, 
\end{align}
where $x_t , u_t, w_t, z_t$ denote the state, control, disturbance, and controlled output. 
In this case, the optimal stabilizing control will be a static state feedback $u_t = -Kx_t$ \cite{kaminer1993mixed}. 

Let $T_{zw}$ denote the input-output map from the disturbance to the measured output. 
We want the $\mathcal{H}_\infty$ norm of the closed-loop system $||T_{zw}||_{\infty}:=\sup \{\frac{||z||_2}{||w||_2}: 0 < ||w||_2 < \infty\}$ to remain below a desired threshold, $\gamma$. 
When $||T_{zw}||_{\infty} < \gamma$, the controller is deemed to have \emph{attenuated} the disturbance\footnote{We will say that the disturbance has been attenuated if $||T_{zw}||_{\infty} < \gamma$ is true even when $\gamma > 1$. In the time-invariant case, $||T_{zw}||_{\infty}$ can be computed as the maximum singular value of the transfer matrix from $w$ to $z$, restricted to the boundary of the unit-circle .}. 

The $\mathcal{H}_2$ norm of a linear system is the expected root mean square value of $z_t$ when $w_t$ is a white noise process \cite{green2012linear}. 
In this case, the $\mathcal{H}_2$-cost will be given by $\lim \limits_{t \rightarrow \infty} \mathbb{E}[z^T_tz_t]$. 
Since $w_t$ in this paper can be a more general adversarial input, we will choose to minimize an upper bound on the (squared) $\mathcal{H}_2$ norm of the closed-loop system \cite{mustafa1991lqg}.
%
\begin{assumption}\label{Assumption1}Assume the following:
\begin{enumerate}
\item $(A, B)$ is stabilizable. This will ensure the existence of a controller $K$ such that $(A-BK)$ is stable. 
\item $C^TC = Q \geq 0$, $E^T C = 0$ and $E^T E = R \geq 0$. This will ensure elimination of cross-weighting terms between state and control variables \cite{mustafa1991lqg}. 
\end{enumerate}
\end{assumption}

Since we are interested in stabilizing controllers that additionally attenuate the adversarial input, we define the \emph{valid} set of controllers as:
\begin{align}
\mathcal{K}:=\{K: |\lambda_{max}(A-BK)| < 1, ||T^K_{zw}||_{\infty} < \gamma\}, \label{ValidController}
\end{align}
where $T^K_{zw}$ is the input-output map from $w$ to $z$ under the controller $K$. 
An upper bound on the infinite-horizon $\mathcal{H}_2$-cost that we seek to minimize in this paper is \cite{haddad1991mixed, mustafa1991lqg}:
\begin{align}
J(K) &:= Tr(PDD^T) \quad \text{ where $P$ solves}\label{Cost}\\
(A-BK)^T &\tilde{P} (A-BK) +Q + K^TRK - P = 0 \label{RiccatiEqn} \\
\tilde{P}&:=P + PD(\gamma^2 I - D^T PD)^{-1} D^T P
\end{align}

A typical objective to achieve mixed $\mathcal{H}_2$/ $\mathcal{H}_\infty$ goals can then be stated as: 
`\emph{For the system in Equations (\ref{StateEqn}) - (\ref{OutputEqn}), determine a sequence of controls $\{u_t\}_{t > 0}$ so that:}
\begin{enumerate}
\item \emph{the cost in Equation (\ref{Cost}) is minimized, subject to}
\item \emph{$u_t = -K x_t$, $K \in \mathcal{K}$, where $\mathcal{K}$ is as in Equation (\ref{ValidController}). }'
\end{enumerate}

If $P \geq 0$ is a solution to Eqn. (\ref{RiccatiEqn}), then $(A-BK)$ is stable if and only if $(A, (Q + K^T R K)^{1/2})$ is detectable \cite{haddad1991mixed}. 
If $P^* \geq 0$ is a solution and $P^* \leq P$ for all other solutions $P$, then $P^*$ is a \emph{minimal} solution. 
The controller that minimizes the cost while achieving $||T^{K^*}_{zw}||_{\infty} < \gamma$ is \cite{haddad1991mixed, zhang2019policy}:
\begin{align}\label{OptCont}
K^*&=(R + B^T \tilde{P}^*B)^{-1} B^T \tilde{P}^*A
\end{align}

We focus on an online setting of Equations (\ref{StateEqn}) - (\ref{OutputEqn}). 
At each time $t$, the adversary chooses $w_t$. 
The learner chooses $u_t = -K_tx_t$, and suffers a loss determined as a function of the matrices $C_t$ and $E_t$. 
We assume that the sequence of matrices $\{C_t, E_t\}$ is determined before the start of the learning process. 
However, they are revealed to the learner only after it chooses $u_t$. 
Therefore, the learner faces a \emph{regret}, defined as the difference between the cost when using the aforementioned controller and the optimal controller from the set $\mathcal{K}$. 
We aim to minimize this regret, and ensure that it grows sub-linearly with the time horizon $T$. 
Formally, 
\begin{problem}\label{ProblemTypical}
At each time $t$, the learner observes state $x_t$, and commits to a controller $u_t = K_t x_t$. 
After this, cost matrices $Q_t:=C_t^T C_t, R_t:= E^T_t E_t$ such that $C_t^T E_t = 0$ are revealed to the learner. 
This cost incurred to the learner is upper-bounded by $J_t (K_t)= Tr (P_t D D^T)$, $P_t$ being the solution of Equation (\ref{RiccatiEqn}). 
With $J_0 (\cdot):=0$, determine a sequence of policies $\{K_t\}$ such that for some large enough time $T$, a regret term, defined as $R(T):= J_T (K_T) -\min_{K \in \mathcal{K}} J_T(K)$ 
grows sub-linearly with $T$. 
\end{problem}
%
%

\section{Solution Method}\label{Solution}

We briefly summarize our solution approach. 
First, we introduce strongly stable disturbance attenuating policies. 
This is motivated by strongly stable policies introduced in \cite{cohen2018online} to quantify the stability of a stabilizing policy. 
Then, we show that if we initialize our procedure with a stable disturbance attenuating policy, successive iterates will continue to yield policies that are stable and disturbance attenuating. 
When solutions of the Riccati recursion are uniformly bounded, we show that the sequence of stabilizing and disturbance attenuating policies are also strongly stable for an appropriate choice of parameters. 
In order to establish our regret bounds, we determine upper bounds on the difference between solutions to the Riccati recursion at successive time steps. 
We put these together to get the final regret bound. 
The regret bound comprises a \emph{burn-in cost}, a cost that is incurred before we start to obtain meaningful bounds, while a second term gives the bound for a large enough time horizon $T$. 
%
\subsection{Strong Stability}

We leverage the notion of a strongly stable controller first proposed in \cite{cohen2018online} for the LQR problem. This was subsequently used in \cite{agarwal2019online} for the more general case. 
\begin{definition}[Strongly Stable Policies \cite{cohen2018online}]
A policy $K$ is stable if $|\lambda_{max}(A-BK)| < 1$. 
It is $(\kappa, \epsilon)-$strongly stable for $\kappa > 0$, $\epsilon \in (0,1]$ if $||K|| \leq \kappa$, and there exist matrices $L, H$ such that $A-BK = HLH^{-1}$, with $||L|| \leq 1-\epsilon$ and $||H||||H^{-1}|| \leq \kappa$. 
\end{definition}

Sequentially strongly stable controllers were used in \cite{cohen2018online, akbari2019iterative} to reason about a sequence of strongly stable policies. 
\begin{definition}[Sequentially Strongly Stable Policies \cite{cohen2018online}]
A sequence of policies $\{K_t\}_{t \geq 1}$ is sequentially $(\kappa, \epsilon)-$strongly stable for $\kappa > 0$, $\epsilon \in (0,1]$ if there exist sequences of matrices $\{L_t\}_{t \geq 1}, \{H_t\}_{t \geq 1}$ such that for all $t \geq 1$,  $A-BK_t = H_t L_t H_t^{-1}$, and:
\begin{enumerate}
\item $||L_t|| \leq 1-\epsilon$, $||K_t|| \leq \kappa$, 
\item $||H_t|| \leq \beta$, $||H_t^{-1}|| \leq 1/\alpha$, where $\kappa = \beta / \alpha$, $\beta > 0$, 
\item $||H_{t+1}^{-1} H_t|| \leq 1+ \epsilon$.
\end{enumerate}
\end{definition}

In the above, observe that $|\lambda_{max}(A-BK_t)| = |\lambda_{max}(L_t)| \leq ||L_t||\leq 1-\epsilon$. 
Since we are interested in stable policies that will also achieve disturbance attenuation, we introduce the notion of strongly stable and sequentially strongly stable disturbance attenuating policies. 

\begin{definition}[Strongly Stable Disturbance Attenuating (S2DA) Policies]
A policy $K$ is $(\kappa, \epsilon, \gamma)-$S2DA if it is $(\kappa, \epsilon)-$strongly stable and $||T^{K}_{zw}||_{\infty} < \gamma$. 
\end{definition}

\begin{definition}[Sequentially Strongly Stable Disturbance Attenuating (S3DA) Policies]
A sequence of policies $\{K_t\}_{t \geq 1}$ is sequentially $(\kappa, \epsilon, \gamma)-$S3DA if $\{K_t\}_{t \geq 1}$ is sequentially $(\kappa, \epsilon)-$strongly stable and $||T^{K_t}_{zw}||_{\infty} < \gamma$ for all $t \geq 1$. 
\end{definition}

\subsection{Set $\mathcal{K}$ is Invariant}

In this part, we will show that if $K_1 \in \mathcal{K}$, then $K_t \in \mathcal{K}$ for all $t > 1$. 
That is, if we start with a stabilizing and disturbance attenuating controller, then successive updates of the controller will retain this property. 
The sequence of controllers is then said to be \emph{regularized} \cite{zhang2019policy}. 
We adapt the Riccati recursion update procedure in \cite{hewer1971iterative} to the setting of disturbance attenuation. 
We use representations of solutions to Lyapunov and Riccati equations to establish stability and disturbance attenuation for the updated controllers. 
Further, since matrices $C_t$ and $E_t$ are fixed at time $t$, we can use results specific to the time-invariant case. 
Before proving our result (Theorem \ref{TheoremSuccStabDist}), we state a useful result from robust control \cite{zhou1996robust} that transforms the constraints in Equation (\ref{ValidController}) to a the solution of a Riccati inequality. 

\begin{lemma}\cite{zhou1996robust}\label{BddRealLemma}
For a discrete-time linear time-invariant system, the following conditions are equivalent:
\begin{enumerate}
\item The controller gain $K \in \mathcal{K}$.
\item There exists $P > 0$ such that: \emph{i)}: $I - \gamma^{-2} D^TPD > 0$, and \emph{ii)}: $Q + K^T R K-P+(A-BK)^T (P + P D(\gamma^2 I - D^T PD)^{-1} D^T P) (A-BK)< 0$.
\item The Riccati equation (\ref{RiccatiEqn}) admits a unique stabilizing solution $P \geq 0$ such that: \emph{i)}: $I - \gamma^{-2} D^TPD > 0$, and \emph{ii)}: $(I - \gamma^{-2} D^TPD)^{-T} (A-BK)$ is stable.
\end{enumerate}
\end{lemma}

In the sequel, for $t \geq 1$, define:
\begin{align}
\bar{R}_t&:=\frac{t-1}{t} \bar{R}_{t-1} + \frac{1}{t}R_t \label{RtUpdate}\\
\bar{Q}_t&:= \frac{t-1}{t} \bar{Q}_{t-1} + \frac{1}{t}Q_t \label{QtUpdate}
\end{align}
We perform the update in this manner in order to obtain useful bounds on differences between successive updates as a function of the time index $t$. 

\begin{theorem}\label{TheoremSuccStabDist}
Let Assumption \ref{Assumption1} hold, $K_1 \in \mathcal{K}$, and there is a solution $P_1 \geq 0$ to Equation (\ref{RiccatiEqn}). 
Suppose at time $t$,
\begin{align}
(A-BK_t)^T \tilde{P}_t &(A-BK_t) +\bar{Q}_t + K_t^T \bar{R}_tK_t = P_t , \nonumber \\
\tilde{P}_t:=P_t + &P_tD(\gamma^2 I - D^T P_tD)^{-1} D^T P_t \label{RiccatiRecn}
\end{align}
and $K_t$ is updated as:
\begin{align}\label{ContrRecn}
K_{t+1}&=(\bar{R}_t + B^T \tilde{P}_t B)^{-1} B^T \tilde{P}_t A. 
\end{align}
Then $K_t \in \mathcal{K}$ for all $t > 1$.
\end{theorem}
\begin{proof}
We will begin by showing that $K_t \in \mathcal{K}$ will ensure that the solution to the Equation (\ref{RiccatiRecn}) is bounded. 
To do this, we use the fact that for a stabilizing $K_t$, the solution to an associated Lyapunov equation will be bounded. 
Then, we will show that $K_{t+1}$ will be stabilizing, and finally show that $K_{t+1}$ will also be disturbance attenuating. 
We use induction. 

\textbf{A.} \emph{\underline{Base Case}}: 

Since $(A, B)$ is stabilizable, there exists a stable controller $K_1$. Since there exists a solution $P_1$ to Equation (\ref{RiccatiEqn}), $K_1$ is also disturbance attenuating (from Lemma 2.1 of \cite{haddad1991mixed}), which establishes the base case of our induction.  

\textbf{B.} \emph{\underline{$P_t$ is Bounded}}: 

Let $K_t \in \mathcal{K}$ for some $t>1$. 
Since $K_t$ is stabilizing, there is a unique solution $\bar{P}_t \geq 0$ to the Lyapunov equation (\ref{LyapEqn}), with $\bar{P}_t$ given by \cite{rugh1996linear}: 
\begin{align}
&(A-BK_t)^T \bar{P}_t (A-BK_t) - \bar{P}_t = -(\bar{Q}_t + K_t^T \bar{R}_t K_t), \label{LyapEqn} \\
&\bar{P}_t=\sum_{i=0}^{\infty} ((A-BK_t)^T)^i (\bar{Q}_t + K_t^T \bar{R}_t K_t) (A-BK_t)^i. \nonumber
\end{align}
Now consider $P_t$ given by Equation (\ref{RiccatiRecn}). 
Subtracting Equation (\ref{LyapEqn}) from Equation (\ref{RiccatiRecn}), we get:
\begin{align}\label{DifferenceRic}
&P_t-\bar{P}_t = (A-BK_t)^T (P_t-\bar{P}_t )(A-BK_t)\\
+&(A-BK_t)^T(P_tD(\gamma^2 I - D^T P_tD)^{-1} D^T P_t)((A-BK_t) \nonumber
\end{align}
Since $K_t \in \mathcal{K}$, from Lemma \ref{BddRealLemma}, $(\gamma^2 I - D^T P_tD) > 0$. 
Therefore, the second term of Equation (\ref{DifferenceRic}) is positive definite, which means that (\ref{DifferenceRic}) is a Lyapunov equation for $P_t-\bar{P}_t$. 
The (unique) solution to this equation is given by: 
\begin{align}\label{DifferenceRicSoln}
&P_t-\bar{P}_t =\sum_{i=0}^{\infty} ((A-BK_t)^T)^i((A-BK_t)^T \times\\&(P_tD(\gamma^2 I - D^T P_tD)^{-1} D^T P_t)((A-BK_t))(A-BK_t)^i. \nonumber
\end{align}
Since $(A-BK_t)$ is stable, both $\bar{P}_t$ and $P_t-\bar{P}_t$ are bounded. 
Therefore, $P_t = \bar{P}_t + (P_t-\bar{P}_t)$ is bounded. 

\textbf{C.} \emph{\underline{$K_{t+1}$ is Stabilizing}}: 

Expanding $(A-BK_t)^T \tilde{P}_t (A-BK_t) + K_t^T \bar{R}_t K_t + \bar{Q}_t$ and using Equation (\ref{ContrRecn}) to write $B^T \tilde{P}_t A = (\bar{R}_t + B^T \tilde{P}_t B) K_{t+1}$:
\begin{align}
&A^T \tilde{P}_t A - K^T_t B^T \tilde{P}_t A - A^T \tilde{P}_t BK_t \nonumber\\&+ K^T_t(B^T \tilde{P}_tB + \bar{R}_t)K_t + \bar{Q}_t \nonumber\\
=&A^T \tilde{P}_t A + (K_{t+1} - K_t)^T(\bar{R}_t + B^T \tilde{P}_t B) (K_{t+1} - K_t) \nonumber\\&-K_{t+1}^TB^T \tilde{P}_t A -A^T \tilde{P}_t B K_{t+1} \nonumber\\&+ K_{t+1}^T (\bar{R}_t + B^T \tilde{P}_t B)K_{t+1}+ \bar{Q}_t \nonumber\\
=&(A-BK_{t+1})^T \tilde{P}_t (A-BK_{t+1}) + K_{t+1}^T \bar{R}_t K_{t+1} \nonumber\\&+ \bar{Q}_t + (K_{t+1} - K_t)^T(\bar{R}_t + B^T \tilde{P}_t B) (K_{t+1} - K_t) \nonumber\\
\Rightarrow P_t &= (A-BK_{t+1})^T P_t (A-BK_{t+1}) + M \label{PtTheorem}
\end{align}
In Equation (\ref{PtTheorem}), $M$ is a positive definite matrix defined as:
\begin{align*}
M&:=K_{t+1}^T \bar{R}_t K_{t+1} + \bar{Q}_t\\+&(A-BK_{t+1})^T (P_t D(\gamma^2 I - D^T P_t D)^{-1} D^TP_t)\\&\qquad \qquad \qquad \qquad \qquad \times(A-BK_{t+1})\\&+ (K_{t+1} - K_t)^T(\bar{R}_t + B^T \tilde{P}_t B) (K_{t+1} - K_t), 
\end{align*} 
The terms in the first and third lines in the above equation are positive definite by assumption, and $(\gamma^2 I - D^T P_t D) > 0$ from Lemma \ref{BddRealLemma}. 
Since $P_t$ is bounded, and we can write \\$P_t = \sum_{i=0}^{\infty} ((A-BK_{t+1})^T)^i M (A-BK_{t+1})^i$, $(A-BK_{t+1})$ must be stable so that the sum on the right hand side does not diverge. 
Therefore, $K_{t+1}$ is stabilizing. 

\textbf{D.} \emph{\underline{$K_{t+1}$ is Disturbance Attenuating}}:

Since $(A-BK_{t+1})$ is stable, there exists $P \geq 0$ that solves $(A-BK_{t+1})^T P (A-BK_{t+1}) - P = -V$, where $V > 0$. 
Choose $V$ to be:
\begin{align*}
&V := K_{t+1}^T \bar{R}_t K_{t+1} + \bar{Q}_t + \rho I \\& \quad + (A-BK_{t+1})^T (P D(\gamma^2 I - D^T P D)^{-1} D^TP)\nonumber\\& \qquad \qquad \times(A-BK_{t+1}),
\end{align*}
where $\rho > 0$ is chosen so that $V$ is positive definite, and
\begin{align}
& \rho I + (A-BK_{t+1})^T (P D(\gamma^2 I - D^T P D)^{-1} D^TP) \nonumber\\& \qquad \qquad \times(A-BK_{t+1})\nonumber\\
&\leq (A-BK_{t+1})^T (P_t D(\gamma^2 I - D^T P_t D)^{-1} D^TP_t) \nonumber\\& \qquad \qquad \times(A-BK_{t+1})\nonumber\\
& \qquad+ (K_{t+1} - K_t)^T(\bar{R}_t + B^T \tilde{P}_t B) (K_{t+1} - K_t). \label{EqnRho}
\end{align}
Rearranging these equations gives us $(A-BK_{t+1})^T \tilde{P} (A-BK_{t+1}) - P +K_{t+1}^T \bar{R}_t K_{t+1} + \bar{Q}_t = -\rho I < 0$, where $\tilde{P}$ is according to Equation (\ref{RiccatiEqn}). 
This satisfies the second part of the second condition in Lemma \ref{BddRealLemma}. 

When $K_t \in \mathcal{K}$, $\gamma^2I-D^T P_t D >0$ from Lemma \ref{BddRealLemma}. 
We can write $
\gamma^2I-D^T P D = \gamma^2I-D^T P_t D + D^T (P_t-P)D$. 
Now, $P_t-P = (A-BK_{t+1})^T (P_t-P) (A-BK_{t+1}) +N$, where $N$ is got by subtracting the term on the left of the inequality in (\ref{EqnRho}) from the term on the right. 
This is a Lyapunov equation in $P_t-P$. 
Since $K_{t+1}$ is stabilizing and $N>0$, there is a positive semi-definite solution, which gives us $P_t-P \geq 0$. 
Therefore, $0 < \gamma^2I-D^T P_t D \leq \gamma^2I-D^T P D$, which satisfies the first part of the second condition in Lemma \ref{BddRealLemma}. 
Then, from Lemma \ref{BddRealLemma}, $K_{t+1}$ is also such that $||T^{K_{t+1}}_{zw}||_\infty < \gamma$, and therefore, $K_{t+1} \in \mathcal{K}$, which completes the proof. 
\end{proof}

\subsection{S2DA and S3DA Policies}

In this part, we present results quantifying the stability and disturbance attenuation of a sequence of valid policies. 
We begin by showing that there exist values of parameters $\kappa, \epsilon$ such that any stable and disturbance attenuating policy is S2DA. 
The proofs are omitted due to space constraints.
%
\begin{proposition}
Assume that $K \in \mathcal{K}$. Then, there exist values $\kappa, \epsilon$ such that $K$ is $(\kappa, \epsilon, \gamma)-$S2DA. 
\end{proposition}
%
 
Suppose that a sequence of positive definite matrices $P_t$ is generated according to Equation (\ref{RiccatiRecn}), where $K_{t+1}$ is given by Equation (\ref{ContrRecn}), and $K_1$ is an initial stable and disturbance attenuating policy. 
Then, we have the following result, assuming that the updates $P_t$ are uniformly bounded. 

\begin{proposition}\label{S2DAPropn}
Let $Q_t ,R_t \geq \mu I$, $P_t \leq \nu I$, and $K_t \in \mathcal{K}$. 
Then, $\{K_t\}_{t \geq 1}$ is $(\bar{\kappa}, \frac{1}{2\bar{\kappa}^2}, \gamma)-$S2DA, where $\bar{\kappa}:=\sqrt{\nu/ \mu}$.
%

Additionally, if $||P_t - P_{t+1}|| \leq p \leq \mu ^2/ \nu$, then, $\{K_t\}_{t \geq 1}$ is $(\bar{\kappa}, \frac{1}{2\bar{\kappa}^2}, \gamma)-$S3DA, where $\bar{\kappa}:=\sqrt{\nu/ \mu}$.
\end{proposition}
%

In the sequel, we will use $\mathcal{K}_{\kappa, \epsilon}$ to denote the set of $(\kappa, \epsilon, \gamma)-$S2DA or $(\kappa, \epsilon, \gamma)-$S3DA policies. 
\subsection{Bound on Riccati Recursion Updates}

Our next result yields a bound on the difference between successive updates of the Riccati recursion (\ref{RiccatiRecn}). 
We achieve this by reducing our framework to the form of the recursive updates for the traditional LQR that was shown in \cite{akbari2019iterative}, and assuming that parameter values are chosen so that an inequality in the proof will not depend on a constant term. 
\begin{theorem}\label{ThmRiccDiff}
Let $Q_t ,R_t \geq \mu I$, $Tr(Q_t), Tr(R_t) < \sigma, P_t \leq \nu I$, and $\{K_t\}_{t \geq 1}$ be $(\kappa, \epsilon, \gamma)-$S2DA. 
Then, there exist constants $p^*$ and $t^*$ such that $||P_{t+1} - P_t|| \leq p^*/t$ for all $t > t^*$.   
\end{theorem}
\begin{proof}
From Equations (\ref{RiccatiRecn}), (\ref{PtTheorem}), and (\ref{ContrRecn}),
\begin{align}
&P_{t+1} - P_t=(A-BK_{t+1})^T  (\tilde{P}_{t+1}  - \tilde{P}_{t})(A-BK_{t+1})\nonumber \\&\quad+ (\bar{Q}_{t+1} - \bar{Q}_t)+ K^T_{t+1} (\bar{R}_{t+1}-\bar{R}_t) K_{t+1} \nonumber\\&\quad- (K_{t+1} - K_t)^T(\bar{R}_t + B^T \tilde{P}_t B) (K_{t+1} - K_t) \\
&K_{t+1} - K_t=
(B^T \tilde{P}_t B + \bar{R}_t)^{-1} \label{EqnContDiff} \\&\quad \times (B^T (\tilde{P}_t - \tilde{P}_{t-1})(A-BK_t) + (\bar{R}_{t-1} - \bar{R}_t)K_t)\nonumber, 
\end{align}
where the last term in the last equation uses the fact that $\bar{R}_{t-1} K_t + B^T \tilde{P}_{t-1} B K_t - B^T \tilde{P}_{t-1} A =0$. 
Therefore, 
\begin{align}
P_{t+1} - P_t&=(A-BK_{t+1})^T (P_{t+1} - P_t)(A-BK_{t+1}) \nonumber \\&\qquad \qquad \qquad \qquad + M_t,  \label{LyapDiffPt}
\end{align}
where $M_t := M_{t_1} + M_{t_2} + M_{t_3}$, and 
\begin{align}
M_{t_1}&:= (A-BK_{t+1})^T \nonumber \\&\times (P_{t+1} D (\gamma^2 I - D^T P_{t+1}D)^{-1} D^T P_{t+1} \nonumber\\ &\qquad- P_{t} D (\gamma^2 I - D^T P_{t}D)^{-1} D^T P_{t})\nonumber \\&\qquad \quad \times (A-BK_{t+1}) \label{Eqn Mt1}\\
M_{t_2}&:=  (\bar{Q}_{t+1} - \bar{Q}_t)+ K^T_{t+1} (\bar{R}_{t+1}-\bar{R}_t) K_{t+1}  \label{EqnMt2} \\
-M_{t_3}&:=(B^T (\tilde{P}_t - \tilde{P}_{t-1})(A-BK_t) + (\bar{R}_{t-1} - \bar{R}_t)K_t)^T\nonumber \\& \qquad \times (B^T \tilde{P}_t B + \bar{R}_t)^{-1} \\& \quad \times (B^T (\tilde{P}_t - \tilde{P}_{t-1})(A-BK_t) + (\bar{R}_{t-1} - \bar{R}_t)K_t) \nonumber
\end{align}
Equation (\ref{LyapDiffPt}) is a Lyapunov equation. Therefore, 
\begin{align*}
P_{t+1} - P_t &= \sum_{i=0}^\infty ((A-BK_{t+1})^T)^i M_t (A-BK_{t+1})^i \\& \leq ||M_t|| \sum_{i=0}^\infty  ((A-BK_{t+1})^T)^i (A-BK_{t+1})^i
\end{align*}
Now, $||M_t|| \leq \sum_{i=1}^3||M_{t_1}||$. 
Since $K_t$ is $(\kappa, \epsilon, \gamma)-$S2DA, $||K_t|| \leq \kappa$, $(A-BK_{t+1}) = H_{t+1} L_{t+1} H_{t+1}^{-1}$, and we have:
\begin{align}
&||\sum_{i=0}^\infty  ((A-BK_{t+1})^T)^i (A-BK_{t+1})^i|| \nonumber\\
&\leq \sum_{i=0}^\infty \kappa^2 (1-\epsilon)^{2i} = \frac{\kappa^2}{\epsilon (2-\epsilon)} \leq \frac{\kappa^2}{\epsilon} \label{ABKBound}
\end{align}
From $Tr(Q_t), Tr(R_t) \leq \sigma$, we can write:
\begin{align}
&||\bar{Q}_{t+1} - \bar{Q}_t|| = \frac{1}{t+1}||Q_{t+1} -\bar{Q}_t|| \leq \frac{2\sigma}{t+1} \nonumber\\
&||\bar{R}_{t+1}-\bar{R}_t||=\frac{1}{t+1}||R_{t+1} -\bar{R}_t|| \leq \frac{2\sigma}{t+1} \nonumber\\
&\qquad\Rightarrow ||M_{t_2}|| \leq \frac{2\sigma (1+\kappa^2)}{t+1} \label{Mt2Bound}
\end{align}
Now, consider $M_{t_1}$. 
We can write $||A-BK_{t+1}|| \leq \kappa (1-\epsilon) \leq \kappa$, since $\epsilon \in (0,1]$\footnote{Note that $|\lambda_{max}(A-BK)| < ||A-BK|| $, where the (two-)norm of a matrix is given by its maximum singular value.}. 
Since $0 < P_t \leq \nu I$, we can write $(\gamma^2 I - D^T P_t D)^{-1} \leq (\gamma^2 I - \nu D^T D)^{-1}$. 
Therefore, 
\begin{align*}
&||(\gamma^2 I - D^T P_t D)^{-1}|| \leq ||(\gamma^2 I - \nu D^T D)^{-1}|| \\
&=||\gamma^{-2} (I - \frac{\nu}{\gamma^2}D^TD)^{-1}|| \leq \frac{1}{\gamma^2} + \frac{\nu}{\gamma^4}||D^T D||
\end{align*}
A lower bound on the norm of the middle term of $M_{t_1}$ is 
\begin{align}
||M_{t_1}|| &\leq  \kappa ^2 m_D, \text{ where} \label{Mt1Bound}\\
m_D:&=\frac{2\nu^2}{\gamma^2} (1+\frac{\nu}{\gamma^2}||D^T D||) ||D^T D|| \nonumber 
\end{align}
Since $R_t \geq \mu I$, $||(B^T \tilde{P}_t B + \bar{R}_t)^{-1}|| \leq \frac{1}{\mu}$. 
Then, we have:
\begin{align}
||M_{t_3}|| \leq \frac{1}{\mu}(\frac{2 \sigma\kappa}{t} + \kappa ||B|| (||P_t - P_{t-1}|| + m_D))^2 \label{Mt3Bound}
\end{align} 
Using the bounds in Equations (\ref{ABKBound})-(\ref{Mt3Bound}), we have: 
\begin{align}
&||P_{t+1} - P_t|| \leq \frac{\kappa^2}{\epsilon}(\kappa^2 m_D + \frac{2 \sigma (1+ \kappa^2)}{t+1}) \nonumber\\
\quad& + \frac{\kappa^2}{\epsilon \mu}(\frac{2 \sigma\kappa}{t} + \kappa ||B|| (||P_t - P_{t-1}|| + m_D))^2 \nonumber\\
=&\frac{2 \kappa^2 \sigma (1+ \kappa^2)}{\epsilon(t+1)} + \frac{\kappa^2}{\epsilon \mu}(\frac{2 \sigma\kappa}{t} + \kappa ||B|| (||P_t - P_{t-1}||))^2 \label{PtRecursiveBound}\\+& \frac{\kappa^4 m_D}{\epsilon}(\frac{2||B||}{\mu}(\frac{2 \sigma}{t}+||B||||P_t - P_{t-1}||)+||B||^2 m_D+1)\nonumber 
\end{align}
To complete the proof, we make the following assumption. 
\begin{assumption}\label{AssumpBound}
$\mu, \nu, \gamma$ are chosen so that the inequality (\ref{PtRecursiveBound}) will be true independent of the last term of (\ref{PtRecursiveBound}) for all $t>t^*$. 
\end{assumption}
%
Future work will examine the relaxation of this assumption in greater detail. 
This setting is now similar to that in Lemma A.6 in \cite{akbari2019iterative}. 
Therefore, if there is some $p^*$ and $t^*$ such that for all $t > t^*$, $(||P_t - P_{t-1}||) \leq p^*/t$, then $(||P_{t+1} - P_{t}||) \leq p^*/(t+1)$. 
Specifically, this will be true for\footnote{These thresholds can be obtained by expanding the quadratic term on the right-hand side of Equation (\ref{PtRecursiveBound}) and using Assumption \ref{AssumpBound} to get a quadratic inequality in $p^*$. 
That is, we get a quadratic $a(p^{*})^2+bp^*+c \leq 0$, where $a, b, c$ are terms involving $t$ and the constants in Equation (\ref{PtRecursiveBound}). 
The bound on $t$ is obtained by recognizing that $(t+1)/t^2 \approx 1/t$, and requiring that the roots of this quadratic inequality be real, that is, $\sqrt{b^2-4ac}>0$. 
The bound on $p^*$ is then got by requiring that $p^* \in [p_1, p_2]$, where $p_1$ and $p_2$ are roots of the quadratic equation $a(p^{*})^2+bp^*+c =0$. 
Specifically, we set $p^* \leq -b/2a$ so that the quadratic inequality will be satisfied.}:
\begin{align*}
t&>t^* = \frac{8 \sigma \kappa^4 ||B||}{\epsilon \mu}(1+ \frac{\kappa^2 ||B|| (1+\kappa^2)}{\epsilon})\\
p^* &\leq \frac{2 \sigma}{||B||} + \frac{4 \kappa^2 \sigma (1+ \kappa^2)}{\epsilon}
\end{align*}
%
The base case of the induction can be shown as in \cite{akbari2019iterative}. 
\end{proof}
%
%

\subsection{Online Algorithm}
\begin{algorithm}[!h]
	\caption{Safety-Critical Online Controller Synthesis}
	\label{algo:RicRecnAlgo}
	\begin{algorithmic}[1]
		\Procedure{Generate $\{K_t\}_{t >1}$}{}
		\State \textbf{Input:} System: $x_{t+1} = Ax_t + Bu_t + D w_t$, initial state, parameters $\mu, \nu, \kappa:=\sqrt{\nu/ \mu}, \epsilon = 1/(2\kappa^2), \gamma, \sigma, K_1 \in  \mathcal{K}_{\kappa, \epsilon}$, time horizon $T$
		\State \textbf{Output:} $\{K_t\}_{t >1}$, such that $K_t \in \mathcal{K}_{\kappa, \epsilon}$
		\For{$t=1,2,\dots, T$}
		\State obtain current state $x_t$
		\State generate $u_t = -K_t x_t$
		\State adversary plays $w_t$
		\State adversary generates $C_t, E_t$ (Assumption \ref{Assumption1})
		\State $Q_t:=C^T_t C_t$; $R_t := E^T_t E_t$
		\State update $\bar{R}_t$, $\bar{Q}_t$ acc. to Eqns. (\ref{RtUpdate})-(\ref{QtUpdate})
		\State update $P_t$ according to Eqn. (\ref{RiccatiRecn})
		\If {$t = \lceil \frac{8 \sigma \kappa^4 ||B||}{\epsilon \mu}(1+ \frac{\kappa^2 ||B|| (1+\kappa^2)}{\epsilon})\rceil$}
		\State $d:=0$, $P_0 := P_{t^*}$, $K_0 := K_{t^*}$
		\State $d \leftarrow d+1$
		\State successively solve Eqn. (\ref{RiccatiRecn}) as long as $||P_d-P_{d-1}|| > p^*/t^*$; update $K_d$ according to Eqn. (\ref{ContrRecn})
		\EndIf
		\State return $K_{t+1}$ according to Eqn. (\ref{ContrRecn})
		\EndFor
		\EndProcedure
	\end{algorithmic}
\end{algorithm}
%
From Assumption \ref{Assumption1} and Theorem \ref{TheoremSuccStabDist}, if we start at $t=1$ from a stabilizing policy that attenuates the disturbance, then our update procedure will continue to yield stabilizing, disturbance attenuating policies for all $t > 1$. 
At each step, we compute $u_t = -K_t x_t$, and the output and cost are revealed in terms of the matrices $C_t$, and $E_t$, where $C_t$ and $E_t$ satisfy Assumption \ref{Assumption1}. 
The update is carried out according to Equations (\ref{RtUpdate})-(\ref{QtUpdate}) by averaging over previous values of $Q_t := C^T_t C_t$ and $R_t:= E^T_t E_t$. 
From Theorem \ref{ThmRiccDiff}, $||P_{t+1}-P_t|| < p^*/t$ for $t > t^*$ (\emph{Lines 12-16}). 
Algorithm \ref{algo:RicRecnAlgo} formally presents this procedure. 
\subsection{Regret Bounds}

Iterative solutions to the Riccati equation in the LTI case exhibit quadratic convergence to an optimal solution $P^*$ \cite{hewer1971iterative}. 
In \cite{zhang2019policy}, this convergence rate was also shown to hold for the variant of the Riccati equation that we use in this paper. 
Specifically, for some $c >0$, $||P_t - P^*|| \leq c ||P_{t-1} - P^*||^2$, and $|| P_{t+1} - P_t|| \leq c ||P_{t} - P_{t-1}||^2$. 
Further, observe that $R(T)$ in Problem \ref{ProblemTypical} can be written as $R(T) = Tr(P_T DD^T) - Tr(P_T^* DD^T)$, where $P^*_T$ corresponds to the solution of the Riccati equation that yields the optimal controller from the set $\mathcal{K}$. 
We use these results to establish a bound on the growth of the regret for sufficiently large $T$. 

\begin{theorem}\label{ThmRegBound}
Let the conditions of Assumption \ref{Assumption1} hold, and let $Q_t ,R_t \geq \mu I$, $Tr(Q_t), Tr(R_t) < \sigma$, $P_t \leq \nu I$, $\kappa = \sqrt{\nu/ \mu},$ $ \epsilon = 1/2\kappa^2$. Let the controllers $\{K_t\}_{t \geq 1} $ be $(\kappa, \epsilon, \gamma)-$S2DA, and $DD^T > 0$. 
Then, for $T \geq t^*=\frac{8 \sigma \kappa^4 ||B||}{\epsilon \mu}(1+ \frac{\kappa^2 ||B|| (1+\kappa^2)}{\epsilon})$, $p^* \leq \frac{2 \sigma}{||B||} + \frac{4 \kappa^2 \sigma (1+ \kappa^2)}{\epsilon}$, and some constant $m>0$, $R(T) \leq Tr (DD^T) (\log (T)+\frac{2mp^*}{t^*+1} - \log (t^*)) $.
\end{theorem}

\begin{proof}
With $J_0(\cdot) = 0$, we can express $R(T)$ as:
\begin{align*}
R&(T)= \sum_{t= 1}^T (Tr(P_t DD^T) - Tr (P_{t-1}DD^T)) - Tr (P_T^* DD^T)\\
&=\sum_{t= 1}^{t^*} (Tr(P_t DD^T) - Tr (P_{t-1}DD^T))- Tr (P_T^* DD^T)\\&+\sum_{t= t^*}^{T} (Tr(P_t DD^T) - Tr (P_{t-1}DD^T))\\
&=Tr(P_{t^*}DD^T) - Tr(P^* DD^T) + Tr (P^* DD^T) \\-& Tr(P_T^* DD^T)+\sum_{t= t^*}^{T} (Tr(P_t DD^T) - Tr (P_{t-1}DD^T))\\
&\leq  2 Tr (DD^T) ||P_{t^*} - P^*|| + Tr (DD^T) \sum_{t= t^*}^{T} ||P_t - P_{t-1}||,
\end{align*}
where $P^*$ is the optimal solution to the time-invariant, infinite-horizon Riccati equation. 
In the above, the first term can be interpreted as a \emph{burn-in cost}, that is, the cost incurred before the procedure starts to yield meaningful regret bounds, while the second term gives the bound for large enough $T$. 

From Theorem \ref{ThmRiccDiff}, $\sum_{t= t^*}^{T} ||P_t - P_{t-1}|| \leq \sum_{t= t^*}^{T} p^* /t \leq \log(\frac{T}{t^*})$, while for the first term, we use the quadratic convergence to $P^*$ to obtain $||P_{t^*} - P^*|| \leq \frac{mp^*}{t^*+1}$. 
Here, $m$ is a constant associated with $\lim \limits_{t \rightarrow \infty} \sum \limits_{i=0}^{t} ||P_{t^*+i} - P_{t^*+i+1}||$. 
%
Therefore, $R(T) \leq Tr (DD^T) (\log (T)+\frac{2mp^*}{t^*+1} - \log (t^*))  .$
\end{proof}

The regret bound in our case differs from those shown in related work (e.g. \cite{akbari2019iterative, agarwal2019online, simchowitz2020improper}) due to the nature of the cost function that we seek to optimize in this paper.  
Since we are interested in the minimization of an (upper bound on the) $\mathcal{H}_2$ cost, given by $\lim \limits_{t \rightarrow \infty} \mathbb{E}[z^T_t z_t]$, when $w_t$ is white noise, our regret term of the form in Problem \ref{ProblemTypical} can be recast in the form on the first line of the above proof.

\section{Experimental Evaluation}\label{Simulation}

We validate our method on a well-studied problem from process control called the Tennessee Eastman control challenge \cite{downs1993plant}. 
The irreversible and exothermic process (Figure \ref{TESchematic}) produces two products ($G, H$) from four reactants ($A, C, D, E$); component $F$ represents other products formed from side reactions in the process.  
The open-loop process is unstable, which necessitates the use of feedback control. 
This model has been adapted to demonstrate the use of machine learning methods to study resilience to attacks \cite{keliris2016machine}, fault detection \cite{zou2018fault}, and impacts of advanced persistent threats \cite{huang2020dynamic}. 
A continuous-time LTI model of the plant presented in \cite{ricker1993model} consisted of eight states, four inputs, and ten outputs. 
We use values of the $A, B, C$ matrices from \cite{ricker1993model}, and discretize the model, assuming a zero order hold. 
We additionally assume $w_t \in \mathbb{R}^8$, $D = I_{8 \times 8}$, and $E$ chosen to satisfy Assumption \ref{Assumption1}. 
We use these values of $C$ and $E$ to determine the (optimal) counterfactual static controller $K \in \mathcal{K}$.
\begin{figure}[!h]
 \centering
  \includegraphics[width=3.85 in]{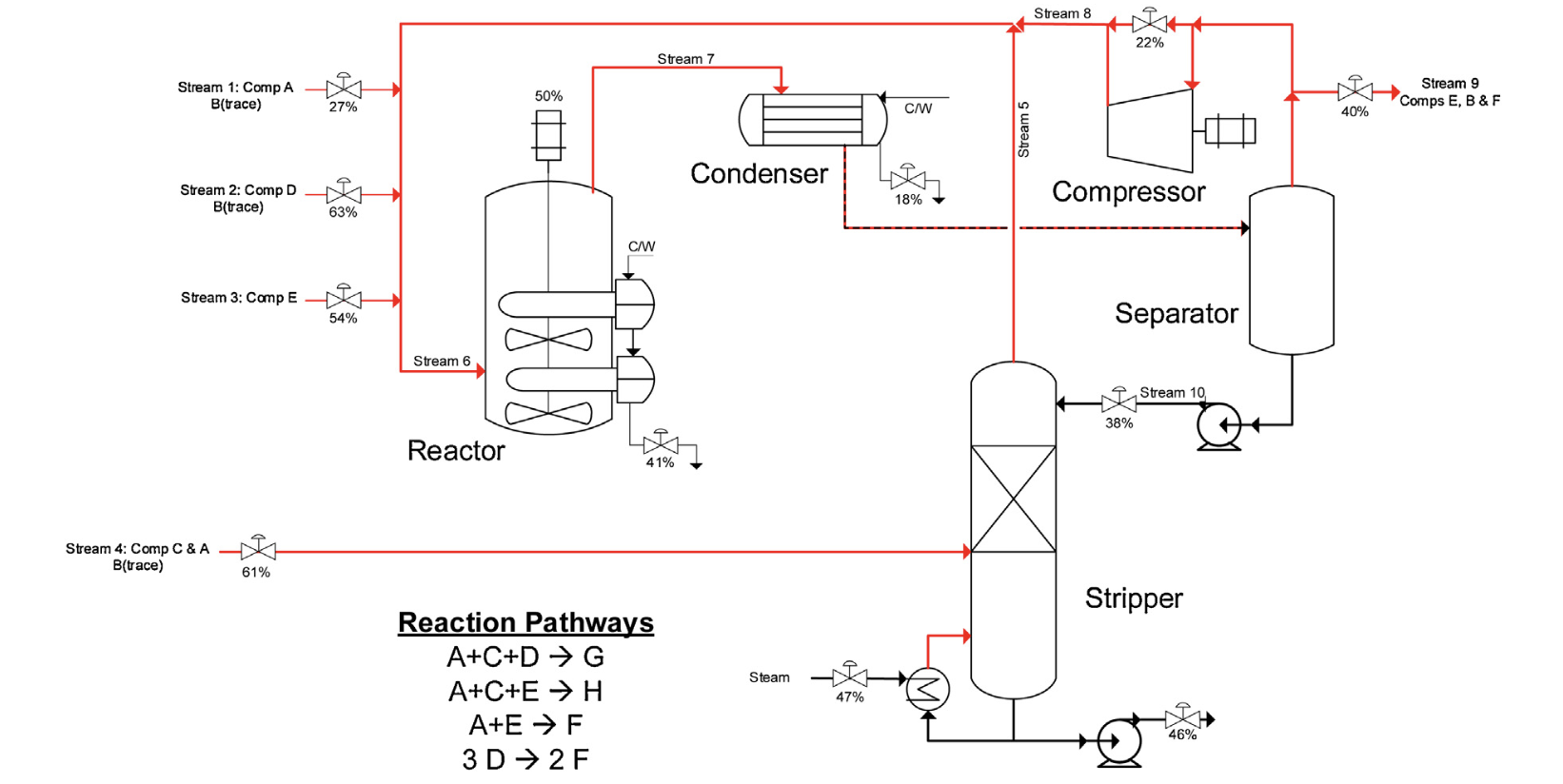} 
\caption{Flow diagram of the Tennessee Eastman Process \cite{udugama2020novel}}\label{TESchematic}
\end{figure}
\begin{figure}[t]
 \centering
 \begin{subfigure}{0.42 \textwidth}
 \centering
  \includegraphics[width = 0.98 \textwidth]{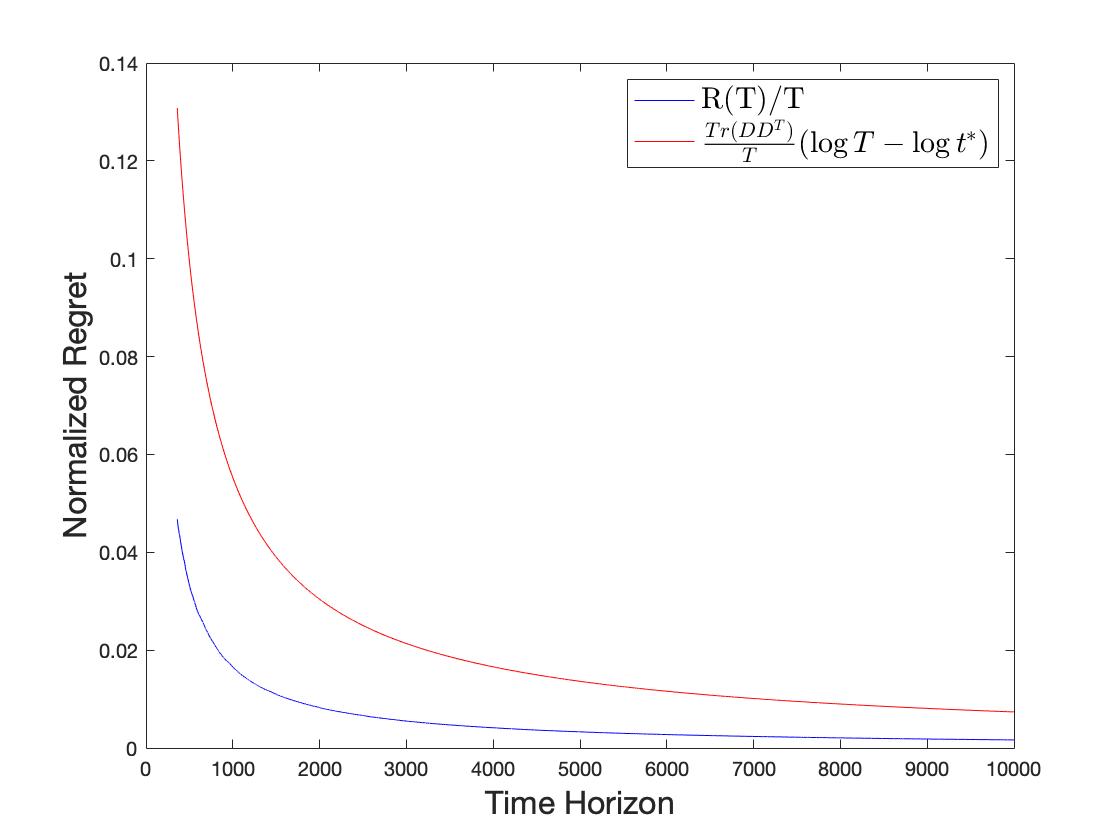} 
\caption{Arbitrary adversarial input.}\label{RegretArbit}
\end{subfigure}
\begin{subfigure}{0.42 \textwidth}
 \centering
  \includegraphics[width = 0.98 \textwidth]{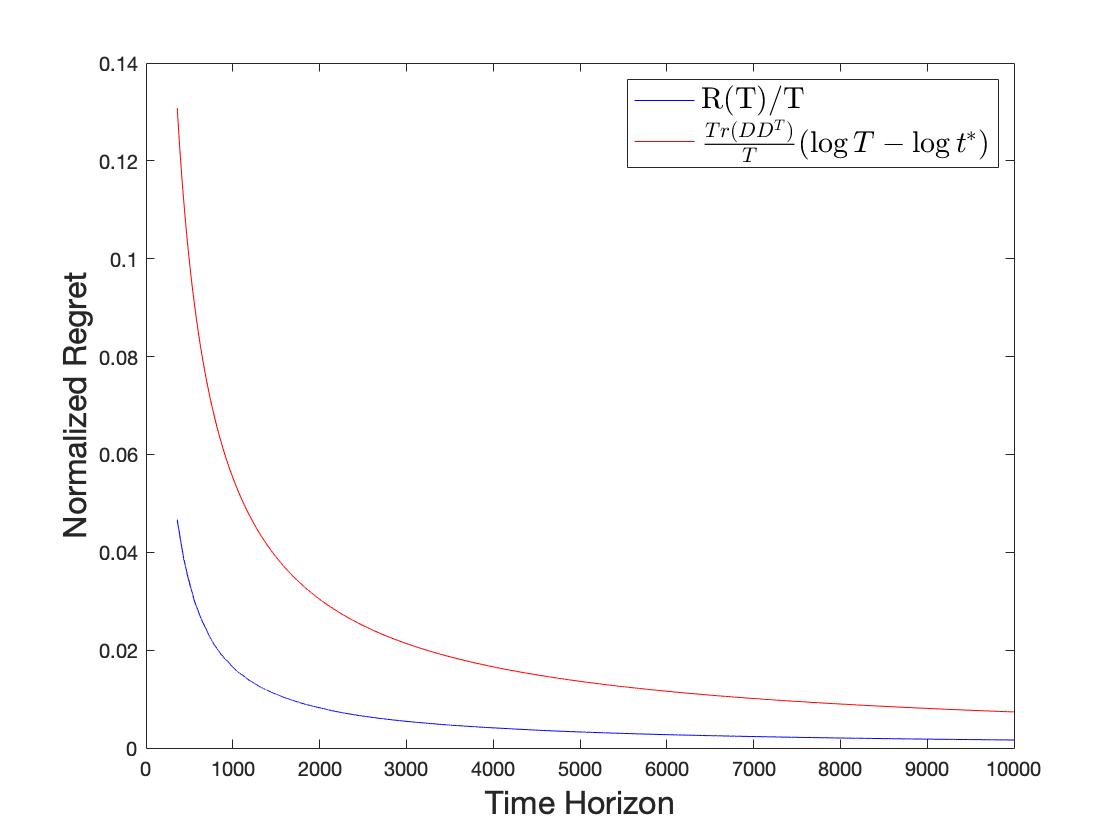} 
\caption{Denial-of-service attack.}\label{RegretDoS}
\end{subfigure}
\caption{Normalized regret for two types of adversarial input. In Fig. \ref{RegretArbit}, the adversary input, $w_t$ is arbitrary. In Fig. \ref{RegretDoS}, $w_t$ cancels the effect of the control $u_t$ on the state evolution during the attack, denoting denial-of-service. The blue curves show the normalized regret of a controller chosen according to Algorithm \ref{algo:RicRecnAlgo} with respect to the optimal, counter-factual, time-invariant static controller. The red curves denote the (normalized) right-hand side of the regret bound of Thm. \ref{ThmRegBound}.}\label{Graphs}
\end{figure}

We consider two attack scenarios. 
In the first, at each time step, the adversary injects an arbitrary input $w_t$. 
We make no assumptions on the nature of this input (except that it is bounded, for the purpose of simulation). 
In the second, we simulate a denial-of-service attack, by setting $w_t = -Bu_t$ for $t \in [t_{a}, t'_{a}]$, and arbitrary at other times. 
$ [t_{a}, t'_{a}]$ is the attack duration such that $1 < t_a \leq t'_a$. 
In this case, the impact of the controller on the evolution of the state is canceled during the attack, but the learner still incurs a cost associated to the control (as a result of the $E_t$ term). 
In each case, the matrices $C_t$ and $E_t$ are perturbed versions of $C$ and $E$. 

The normalized regret for the two attacks are shown in Figure \ref{Graphs}. 
In particular, we observe that the regret of a sequence of controllers computed according to Algorithm \ref{algo:RicRecnAlgo} with respect to the optimal, counter-factual, time-invariant static controller that is obtained by solving the Riccati equation for the time-invariant case satisfies the bounds determined in Theorem \ref{ThmRegBound}. 
For the denial-of-service attack, although the effect of the controller is canceled for the duration of the attack, as long as this attack starts at $t_a >1$, Algorithm \ref{algo:RicRecnAlgo} will continue to produce stabilizing, disturbance attenuating controllers if we start from an initial controller that is stabilizing and disturbance attenuating (Theorem \ref{TheoremSuccStabDist}).
%
%
\section{Conclusion}\label{Conclusion}

This paper presented an iterative solution to an online control problem in the presence of bounded adversarial disturbances. 
In this setting, costs incurred by the system at each time due to an adversarial disturbance input were revealed only after the input was given. 
We synthesized controllers to minimize (an upper bound of) a quadratic cost while simultaneously satisfying a safety constraint. 
This was achieved by solving a Riccati equation in an iterative manner. 
Solutions to the Riccati equation enforced the safety constraint. 
We showed that initializing the procedure with a stabilizing and disturbance attenuating controller ensured that controllers at successive time steps retained this property. 
We showed that the regret of this controller, compared to the optimal controller when all costs and disturbances were known in hindsight, varied logarithmically with the time horizon. 
We validated our approach on a model of the Tennessee Eastman chemical process that was subject to arbitrary adversarial inputs and a denial of service attack. 

Future work will 
study the partial information setting, where one will have to synthesize dynamic output feedback controllers, and the more generalized problem of minimizing the $\mathcal{H}_{\infty}$ norm of the output to disturbance map. 
We will also extend our analysis to the case of unknown system dynamics.
%

\bibliographystyle{IEEEtran}
\bibliography{RLDDPRef}
\end{document}